\documentclass[12pt,a4paper]{article}
\usepackage{graphicx}
\usepackage{amssymb,amsmath,amsthm}
\newcommand{\C}{\mathbb{C}}
\newcommand{\N}{\mathbb{N}}
\newcommand{\R}{\mathbb{R}}

\newcommand{\Hm}[1]{\leavevmode{\marginpar{\tiny%
$\hbox to 0mm{\hspace*{-0.5mm}$\leftarrow$\hss}%
\vcenter{\vrule depth 0.1mm height 0.1mm width \the\marginparwidth}%
\hbox to
0mm{\hss$\rightarrow$\hspace*{-0.5mm}}$\\\relax\raggedright #1}}}
\newtheorem{claim}{Claim}[section]
\newtheorem{theorem}[claim]{Theorem}
\newtheorem{lemma}[claim]{Lemma}

\theoremstyle{definition}

\begin{document}


\title{\bf A straight waveguide with a wire inducing resonances}

\author{Sylwia Kondej$^a$ 
and Wies{\l}aw Leo\'{n}ski$^b$ }
%


\medskip

\maketitle

\medskip
\noindent a) Institute of Physics, University of Zielona G\'ora,
ul.\ Szafrana
 4a, 65246 Zielona G\'ora, Poland; S.Kondej@if.uz.zgora.pl
\medskip

 \noindent b) Institute of Physics, University of Zielona G\'ora, ul.\ Szafrana
 4a, 65246 Zielona G\'ora, Poland; W.Leonski@if.uz.zgora.pl

\medskip

\begin{abstract} We study a straight infinite planer waveguide
with, so called, leaky  wire attached to the walls of the
waveguide. The wire  is modelled by an attractive delta
interaction supported by a finite segment. If the  wire is placed
perpendicularly then the system preserves mirror symmetry which
leads the embedded eigenvalues phenomena. We show  that if we
break the symmetry the corresponding resolvent poles turn to
resonances. The widths of resonances are calculated explicitly in
the lowest order perturbation term.
\end{abstract}
%
%
\bigskip
{\bf key words:} straight Dirichlet waveguide, delta interaction,
embedded eigenvalues, resonances
\section{Introduction}
%
\emph{A wavequide with a leaky wire.} The paper belongs to the
line of  research often called Schr\"{o}dinger operator with delta
interactions. To explain the physical motivations let us consider
a quantum particle moving in a straight planar waveguide $\Omega
:= \{(x_1\,,x_2 ) \in \R^2 \,:\, x_1 \in \R\,, x_2 \in (0,\pi) \}
$ of the width $\pi$. Assuming that the waveguide forms
impenetrable walls at its boundaries we exclude a possibility of
tunnelling apart of the area $\Omega $; mathematically it
corresponds to the Dirichlet boundary conditions at $\partial
\Omega $. Moreover, the interface of two different materials
produces an additional jump of potential localized on a finite
line $\Sigma $ placed in the waveguide. In fact, $\Sigma$ is
defined  by a finite straight line attached to the walls of
$\Omega $, precisely  $\Sigma \equiv \Sigma _\epsilon :=\{
(\epsilon x_2 \,, x_2 )\in \R^2 \,, x_2 \in (0,\pi )\}$. If
$\epsilon =0$ then the potential is localized perpendicularly to
$\partial \Omega $ and, consequently, the system possess the
mirror symmetry with respect to the axis $O=\{(x\,, \pi /2)\,:\,
x\in \R\}$. Therefore $\epsilon$ determines a breaking symmetry
parameter. In fact,  positivity and negativity of $\epsilon$ leads
to  the same physical model. Therefore, without loosing a
generality, we assume $\epsilon \geq 0$.
Since the potential is supported by the set $\Sigma_\epsilon $ of
the lower dimension
 it imitates, so called, leaky quantum wire
which can be modelled  by delta interaction. A large number of
models with delta interaction is discussed in \cite{AGHH}. Some
specific models of leaky  quantum wires were studied in \cite{BT,
EI, EK-strong, EK1, EY1, EP, K12, KK13}.

\medskip

\noindent \emph{Hamiltonian of the  system } The above described
system is governed by the Hamiltonian which can be symbolically
written as
\begin{equation}\label{eq-Ham}
  -\Delta _\Omega -\alpha \delta (\cdot -\Sigma_\epsilon
  )\,,\quad \alpha \in \R\,,
\end{equation}
where $\Delta _\Omega $ is the Laplace operator acting in $L^2 (
\Omega )$ with Dirichlet boundary conditions, $\delta (\cdot
-\Sigma_\epsilon
  ) $ stands for the Dirac delta modelling the wire and $\alpha $
  describes the interaction supported by $\Sigma _{\epsilon}$. In
  the following we mainly assume that the interaction is
  attractive. In this case $\alpha $ is positive. To give a
  mathematical meaning to the above formal expression we use the
   form-sum method. More precisely, we construct the form which
  intuitively preserves the properties of (\ref{eq-Ham}) and, at
  the same time, the operator $H_{\alpha , \epsilon }$ associated to this form
  is self-adjoint.

\medskip

\noindent \emph{The main results of the paper.} The model is
attractive not only due to the applications to quantum mechanics.
Mathematically, it reveals a relation between geometry of quantum
system and its spectral properties. \emph{Our main aim is to
recover how the geometry of wire represented by the parameter
$\epsilon $ affects the spectral properties of the Hamiltonian
$H_{\alpha , \epsilon}$.}

To get insight into the spectral structure of the model at hand
let us start our analysis form the case $\epsilon =0$. Then due
the symmetry of the system the Hamiltonian can be decomposed onto
longitudinal component defined by the one dimensional Laplace
operator $-\Delta^{(1)}$ with one point interaction determined the
coupling constant $\alpha >0$ and transversal component determined
by the Dirichlet Laplacian acting in $L^2 ([0\,,\pi ])$. The
latter has a quantized spectrum given by $\{k^2\}_{k\in \N}$.
Consequently, the original Hamiltonian $H_{\alpha, 0}$ is unitary
equivalent to the orthogonal sum of the operators  which can be
symbolically written as
\begin{equation}\label{eq-sum}
\oplus_{k\in \N} \,H^k_{\alpha }\,, \quad H^k_{\alpha }= -\Delta
^{(1)}-\alpha \delta (x)+k^2 \,.
\end{equation}
 The spectrum of the particular component is given by
\begin{equation}\label{eq-spectrum0}
\sigma (H^k_\alpha )= [k^2\,,\infty) \cup\{E_k\}\,, \quad k\in
\N\,,\end{equation}
 where
$$
E_k = -\frac{\alpha ^2}{4}+k^2\,;
$$
cf.~\cite{AGHH}. The threshold of the essential spectrum of
$H_{\alpha , 0}$ is determined by the ground state of $H^1_\alpha
$. This means that $\sigma_{\mathrm{ess}}(H_{\alpha
,0})=[1,\infty)$. Moreover, the same  half line determines the
essential spectrum of the system if $\epsilon
>0$, i.e. we have
$ \sigma _{\mathrm{ess}}(H_{\alpha , \epsilon}) = [1,\infty )$,
see~Section~\ref{sec-spectrum}. Furthermore, it follows from
(\ref{eq-spectrum0}) that
  $$
\sigma _{\mathrm{p}} (H_{\alpha , 0})= \{E_k\}_{k\in \N}\,.
  $$
Therefore, the wire induces \emph{at least one discrete point} of
the spectrum below the threshold $1$ and \emph{infinite number of
the embedded eigenvalues}. The latter correspond to the phenomena
of autoionization. The states located above ionization threshold
can be experimentally observed, for instance as bumps in the
scattering cross section when an electron leaves multi-electron
atom. Usually helium or more often, barium atoms are applied in
such experimental setups. Autoionizing systems are often described
by means of, so called, \textit{Fano profiles}
\cite{F61,LTK87,DP02}, and they were discussed in various
contexts. For instance see \cite{LB90,L93} (and the references
quoted therein), where various models of laser-atomic system
interactions were considered.

In our system the gaps between subsequent embedded eigenvalues
behaves as
$$
E_{k+1} -E_k = 2k +1\,.
$$
 We show that for small parameter $\epsilon >0$
  the embedded eigenvalues are recovered from the essential spectrum and  move
  to the second sheet continuation of the resolvent, i.e. they produce
 resonances. Furthermore, it is proved that for $E_n >1$ the pole of the
 resolvent is localized at
 $$
z_n =E_n + V_n \epsilon +W_n \epsilon^2
+\mathcal{O}(\epsilon^3)\,,
 $$
where $V_n$ and $W_n$  are found out explicitly. We show that
$V_n$ is real so it does not contribute to  the width of
resonance. Therefore $V_n $ states the first perturbation term of
the ``energy shift''. The lowest order of the resonance widths
$\gamma_n$ can be recovered from the imaginary component of $W_n$;
precisely given $n\in \N$ there exists $\epsilon_n$ such that for
$\epsilon \in (0\,,\epsilon_n )$ we have
$$ \gamma_n = 2 \, |\Im
W_n |\epsilon^2\,.$$
\\
Finally, let us mention that the problem of a waveguide with delta
interaction was studied in~\cite{EK-1}. The authors discussed
spectral properties of the system with varying longitudinal
straight line interaction. From the spectral point of view the
mentioned  model is essentially different from waveguide with the
``almost perpendicular'' wire which we analyzed in this paper. For
the related problems concerning spectral properties  in the
straight waveguides we recommend  \cite{AKK,BC1, BC2, EN, EK-2}.
On  the other hand the resonance models with delta interactions
were studied in  \cite{EK-resonances, K12, KK13}.
\bigskip

\noindent \emph{Notations.}

\medskip

\noindent $\bullet$ We abbreviate $L^2_\Sigma \equiv L^2
(\Sigma\,, \mathrm{d}l)$, where $\mathrm{d}l$ is a linear measure
on $\Gamma$; the corresponding scalar product will be denoted
$(\cdot, \cdot )_{\Sigma }$.

\medskip

\noindent $\bullet$ In the following the space $L^2
([0,\pi])\equiv L^2 ([0,\pi]\,, \mathrm{d}x_2 )$ is frequently
used; the corresponding scalar product will be shortly denoted as
$(\cdot, \cdot )$.

\medskip

\noindent  $\bullet $ The notation  $\Delta $ stands for the two
dimensional Laplace operator acting in $L^2(\Omega )$ and
satisfying the Dirichlet boundary conditions.

\medskip

\noindent  $\bullet$ We use the standard notation
$W^{n,2}_0(\Omega )$, $n\in\N$, for the Sobolev space of the trace
zero functions.

\medskip

\noindent $\bullet $ We denote  $ \omega_k (\cdot ) := \sin
k(\cdot) $.

\section{Hamiltonian of the system and its resolvent}

\subsection{Hamiltonian}
\emph{The unperturbed Hamiltonian.} Let us start our discussion
with the  ``free'' Dirichlet waveguide. The Hamiltonian of such
system is given by the self-adjoint operator
$$
H_0 =-\Delta\,:\, \mathcal{D} (H_0) \to L^2 (\Omega )\,,
$$
where the domain $\mathcal{D} (H_0)$ coincides with
$W^{2,2}_0(\Omega )$. Operator  $H_0$ is associated with the
following quadratic form
$$
\mathcal{E}_{0} [f] = \int_{\Omega }|\nabla f|^2 \,, \quad f\in
W_0^{1,2}(\Omega )\,.
$$
The translational symmetry of the system leads to the following
decomposition 
\begin{equation}\label{eq-decom}
H_0 = -\Delta^{(1)}\otimes 1 + 1 \otimes -\Delta^{(1)}_D \,,\quad
\mathrm{onto}\quad L^{2} (\R)\otimes L^2 ([0,\pi])\,,
\end{equation}
where $\Delta^{(1)}$ is the one-dimensional Laplace operator
acting in $L^2 (\R)$ and $-\Delta^{(1)}_D$ is the one-dimensional
Laplace operator acting in $L^2 ([0,\pi])$ and satisfying
Dirichlet boundary conditions.

\noindent \emph{The Hamiltonian of the waveguide with a wire.} Now
we introduce  delta potential localized at $\Sigma_\epsilon$. The
modified Hamiltonian formally corresponds to  (\ref{eq-Ham}). To
give a mathematical meaning to this  formal expression we employ
the form-sum method. Consider
 the sesqilinear form
 $$
\mathcal{E}_{\alpha , \epsilon} [f] = \int_{\Omega }|\nabla f|^2
-\alpha \int_{\Sigma_\epsilon} |I_{\epsilon} f|^2 \,, \quad f\in
W_0^{1,2} (\Omega )\,,
 $$
where $I_\epsilon $ stands for the continuous embedding operator
of $W_0^{1,2} (\Omega )$ into $L^2 (\Sigma_\epsilon) $. The form
$\mathcal{E}_{\alpha , \epsilon} $ is symmetric and semi-bounded.
Moreover, the Dirac delta supported by $\Sigma_\epsilon$ defines
the Kato class measure,~cf.~\cite{BEKS}. This implies the
closeness of $\mathcal{E}_{\alpha , \epsilon} $. Consequently,
there exists uniquely defined operator $H_{\alpha , \epsilon }$
associated with the form $\mathcal{E}_{\alpha , \epsilon} $ via
the second representation theorem. This operator defines the
Hamiltonian of our system.
\\ \\
\subsection{Resolvent of $H_{\alpha , \epsilon}$}
\noindent \emph{Resolvent of ``free'' Hamiltonian.} The
Hamiltonian of the unperturbed waveguide does not admit any bound
states. The threshold of the essential spectrum is determined by
the lowest transversally quantized energy, i.e. to the
  ground state energy of
$-\Delta_D^{(1)}$,~cf.~(\ref{eq-decom}). Since the waveguide width
equals  $\pi$ the lowest transversal energy is $1$. Therefore
$$
\sigma(H _{ 0})=[1, \infty )\,.
$$
Suppose $z\in \C \setminus [1, \infty )$ and $R_0(z)$ stands for
the resolvent of $H_0$, i.e. $R_0(z)=(H_0 - z )^{-1}$. Then
$R_0(z)$ is an integral operator with the kernel
\begin{equation}\label{eq-freekernel}
G_0 (z; x, y )= \frac{i}{\pi } \sum_{k=1}^{\infty}
\frac{\mathrm{e}^{i\sqrt{z-k^2}|x_1 -y_1|}}{\sqrt{z-k^2}} \omega_k
(x_2) \omega_k (y_2 )\,,
\end{equation}
where $x=(x_1\,,y_2)$ (analogously  $y$), the square root function
is defined on the first Riemann sheet and $ \omega_k (\cdot ) :=
\sin k (\cdot )$.
\\
 \noindent \emph{Resolvent of $H_{\alpha, \epsilon}$.} To reconstruct
 the Krein like resolvent of $H_{\alpha, \epsilon}$ we have to introduce
 embeddings of the  ``free'' resolvent to the space $L^2
 (\Sigma_\epsilon)$. By means of the trace map  $I_{\epsilon}$ we define
$$
\tilde{R}_{\epsilon } (z) := I_{\epsilon } R_0 (z)\,:\, L^2
(\Omega )\to L^2 (\Sigma_\epsilon)\,.
$$
Its adjoint $\tilde{R}^\ast _{\epsilon } (z)\,:\, L^2
(\Sigma_\epsilon )\to L^2 (\Omega )$ acts as $\tilde{R}^\ast
_\epsilon (z)f= G_0 (z)\ast f\delta _{\Sigma _\epsilon}$.
Moreover, we introduce the bilateral embedding
$$
\mathrm{R}_{\epsilon }(z)= I_{\epsilon} \tilde{R}^\ast _{\epsilon
}(z)\,:\, L^2 (\Sigma_\epsilon) \to L^2 (\Sigma_\epsilon )\,.
$$
Finally, we define the operator which is a key tool  for further
spectral analysis based on the generalized Birman-Schwinger
argument. Let
$$
\Gamma _\epsilon (z):= I -\alpha \mathrm{R}_{\epsilon }(z)\, :\,
L^2 (\Sigma _\epsilon ) \to  L^2 (\Sigma _\epsilon ) \,.
$$
Relaying on the results of \cite{BEKS} we can formulate the
following statement.
\begin{theorem} Suppose $z\in \C\setminus [1\,,\infty)$,  and operator
$\Gamma _\epsilon (z)$ is invertible. Then the operator
\begin{equation}\label{eq-resolvent}
  R_{\alpha, \epsilon}(z) = R_0 (z) +\tilde{R}^\ast _{\epsilon}(z)
  \Gamma_\epsilon
  (z)^{-1}\tilde{R} _{\epsilon} (z)\,
\end{equation}
defines the resolvent of $H_{\alpha, \epsilon}$, i.e. $R_{\alpha
,\epsilon  }(z)= (H_{\alpha , \epsilon} -z )^{-1}$.\\
Moreover
\begin{equation}\label{eq-BS}
  \dim \, \ker (H_{\alpha ,\epsilon }-z)  =  \dim \, \ker \Gamma _\epsilon
  (z)\,
\end{equation}
and the map $h\mapsto \tilde{R}^\ast _{\epsilon}(z)h$ defines a
bijection from $ \ker \Gamma _\epsilon
  (z)$ onto $\ker (H_{\alpha ,\epsilon }-z)$.
\end{theorem}
The above theorem allows to ``shift'' the eigenvalue problem for
the differential operator $H_{\alpha, \epsilon}$ to the problem of
zeroes of the integral operator $\Gamma _\epsilon (\cdot )$. \\
For the more general approach we recommend \cite{P01}.

\subsection{ Reparameterization} A certain disadvantage
of the operator $\Gamma_\epsilon (z)$ is the fact that it acts in
the trace space $L^2 (\Sigma_\epsilon )$ which is $\epsilon $
dependent. To make it ``stable'' w.r.t. $\epsilon$ we parameterize
$\Sigma _\epsilon $ by means of $x_2$ instead of the length of arc
$l$. Recall $\Sigma _\epsilon =\{ ( \epsilon x_2 , x_2 ) \in
\Omega \,:\, x_2 \in (0, \pi  )\} $. Using the relation
$\frac{dl}{dx_2}=(\epsilon^2 +1)^{1/2}$ we conclude that the map
$$f \mapsto (\epsilon^2 +1)^{1/4}f ( (\epsilon^2 +1)^{1/2} (\cdot
))$$  defines a unitary operator $U\,:\, L^2 (\Sigma _\epsilon)
\to L^2 ([0,\pi])$. \\
Without a danger of  confusion we keep the same notation
$\mathrm{R}_\epsilon (z)$ for its unitary ``equivalent'' $U
\mathrm{R}_\epsilon (z)U^{-1}\,:\, L^2 ([0,\pi])\to L^2 ([0,\pi])$
and analogously for $\Gamma _\epsilon (z) $.
\\
After the re-parameterization the kernel of the operator
$\mathrm{R}_{\epsilon }(z)\,:\,L^2 ([0,\pi]) \to L^2 ([0,\pi])$ is
given by
\begin{equation}\label{eq-Gepsilon}
 \mathrm{G}_\epsilon (z; x_2,y_2)=\frac{i}{\pi }
\sum_{k=1}^{\infty} \frac{\mathrm{e}^{i\sqrt{z-k^2}\epsilon |x_2
-y_2|}}{\sqrt{z-k^2}} \omega _k (x_2 )\omega_k (y_2)\,,
\end{equation}
cf.~(\ref{eq-freekernel}).

\section{Spectrum of  $H_{\epsilon, \alpha }$; embedded eigenvalues phenomena }
\label{sec-spectrum}

\emph{Stability of the essential spectrum.} Assume $\alpha
>0$. Since the delta perturbation is compactly supported one may
expect the stability of the essential spectrum with respect to the
"free" Hamiltonian $H_0$.  It was shown in \cite{BEKS} that  a
finite measure preserves the essential spectrum of the Laplacian
acting in $L^2 (\R^n)$. The argument employed in \cite{BEKS} can
be extended for the Laplacian acting in $L^2 (\Omega )$. Namely,
the both operators $\tilde{R}^\ast _\epsilon (z)$ and
$\tilde{R}_\epsilon (z)$ are compact. Indeed,  note that the
operator $R_0 (z)\,:\, L^2 (\Omega )\to W^{2,2}_0 (\Omega )$ is
bounded. Furthermore, using the standard embedding theorem,
cf.~\cite{A}, we  conclude that  the inclusion $W^{2,2}_0 (\Omega
) \subset C^0 (\Sigma_\epsilon )$ defines a compact embedding
operator in the sense of the trace map. On the other hand, since
the embedding $C^0 (\Sigma_\epsilon) \subset L^2
(\Sigma_\epsilon)$ is continuous we conclude that the trace map
$I_\epsilon \,:\, W^{2,2}_0 (\Omega ) \subset W^{1,2}_0 (\Omega )
\to L^2 (\Sigma_\epsilon)$ is compact. This yields compactness of
$\tilde{R}_\epsilon (z)$ as well as its adjoint $\tilde{R}^\ast
_\epsilon (z)$.
 Consequently, the continuity
 of $\Gamma_\epsilon (z)^{-1}$ implies compactness of
$\tilde{R}^\ast _{\epsilon}(z)
  \Gamma_\epsilon
  (z)^{-1}\tilde{R} _{\epsilon} (z)
  $. In view the Weyl theorem we have the stability of the
essential spectrum
$$
\sigma _{\mathrm{ess}} (H_{\alpha , \epsilon })=\sigma
_{\mathrm{ess}} (H_0)=[1,\infty)\,.
$$

\noindent \emph {Point spectrum of $H_{\alpha , 0}$; poles of the
resolvent.} Our first aim is to characterize the point spectrum of
the Hamiltonian $H_{\alpha , 0 }$ which governs the system with
the wire placed perpendicularly to the boundaries of $\Omega $. In
fact, the problem was solved in \cite{AKK} employing the
generalized sum method. Moreover, relaying on the symmetry
argument we provide a characterization of the spectrum of
$H_{\alpha , 0}$ in~Introduction. Now we would like to find out
the point spectrum of $H_{\alpha , 0 }$ as the poles of its
resolvent.
\\
Using (\ref{eq-Gepsilon}) we have
\begin{equation}\label{eq-G}
\mathrm{G}_{0} (z; x_2 ,y_2) =\frac{i}{\pi} \sum _{k=1}^\infty
\frac{1}{\sqrt{z-k^2}} \omega_k (x_2) \omega_k (y_2)\,.
\end{equation}
Since the sequel $\{ \left( \frac{2}{\pi }\right)^{1/2} \omega_k
(\cdot )\}_{k\in \N } $ forms an orthonormal basis in $L^2
([0,\pi])$ we arrive at
\begin{equation}\label{eq-BSII}
\ker \Gamma _{0} (z) \neq \{ 0\}
\,\,\,\,\Longleftrightarrow\,\,\,\, 2 \sqrt{z-k^2} - i\alpha
=0\,,\,\, k\in \N\,.
\end{equation}
Note that the latter admits solutions also  for $z\in
[1\,,\infty)$ and $\alpha >0$. Precisely, given $k\in \N$ the
number
\begin{equation}\label{eq-ev}
E_k = -\frac{\alpha ^2}{4} +k^2\,
\end{equation}
determines  a solution of (\ref{eq-BSII}) and, at the same time,
determines an eigenvalue of $H_{\alpha , 0 }$,~cf.~\cite{AKK}. In
the following theorem we summarize the above discussion.
\begin{theorem} Assume $\alpha >0$. The  eigenvalues of $H_{\alpha , 0}$ take the
form (\ref{eq-ev}). Consequently, the number of the discrete
spectrum points is finite and  given by
$$
\sharp \sigma _{\mathrm{d}} (H_{\alpha , 0})=  \{ k\in \N \,:\,
-\frac{\alpha ^2}{4}+k^2 <1 \}\,.
$$
On the other hand, for any  $k\in \{ n\in \N \,:\, -\frac{\alpha
^2}{4}+k^2 \geq 1 \}$ $E_k$ given by (\ref{eq-ev}) constitutes an
embedded eigenvalue of $H_{\alpha , 0}$. Therefore, the number of
embedded eigenvalues  is infinite.
\end{theorem}

\section{Resonances}
\subsection{Resonances as the poles of the resolvent} The
Birman--Schwinger principle  is based on a relation between the
eigenvalues of the Hamiltonian and the poles of its resolvent.
Using this argument we have stated that the Hamiltonian $H_{\alpha
, 0}$ admits the embedded eigenvalues which determine  zeros of
$\Gamma_0 (\cdot)$.  It arises the question: where are  poles of
the resolvent localized if we ``slightly'' break  the symmetry?
The aim of this section is to show that the resolvent of
$H_{\alpha, \epsilon }$ admits  the second sheet analytical
continuation in a certain sense and, moreover, the second sheet
continuation $\Gamma _\epsilon ^{\mathit{II} }(\cdot )$ of $\Gamma
_\epsilon (\cdot )$ has zeroes in the lower half plane,~i.e.
\begin{equation}\label{eq-2sheet}
  \ker \Gamma _\epsilon ^{\mathit{II} }(z) \neq  \{0\}\,,\quad \Im z <0\,.
\end{equation}
The resolvent poles defined by (\ref{eq-2sheet}) constitute
resonances.

\subsection{Analytical continuation of $\mathrm{R}_{\epsilon }(z)$}

Let  $n\in \N$ and  $E_n= -\frac{\alpha ^2}{4}+n^2$ be an embedded
eigenvalue of $H_{\alpha, 0}$. Henceforth we assume $\alpha \neq 2
\sqrt{n^2 -k^2}$ for any $k \in \N$ such that $k \leq n$. In this
way we exclude the case when $E_n $ reaches the threshold of the
essential spectrum of $H^k_{\alpha }$;~cf.~(\ref{eq-sum}).
\\
The goal  of this section is to construct the analytical
continuation of $\mathrm{R}_{\epsilon }(\cdot )$ in a neighborhood
of $E_n$. In fact, $E_n$ is localized between the neighboring
thresholds of the essential spectra of (\ref{eq-sum}). These
thresholds determine the boundaries of the largest real interval
containing $E_n$ and admitting an analytical continuation of
$\mathrm{R}_{\epsilon }(z)$ to the lower half plane. To specify
this interval let us define
$$
\mathcal{A}_n:=\{k\,:\, k^2 < E_n\}\,,\quad \mathcal{A}_n^c := \N
\setminus \mathcal{A}_{n}\,.
$$
Suppose $k_1 := \max\{k\in \N \,:\, k\in \mathcal{A}_n\} $ and set
$\Upsilon _n:= (k_1^2 \,, (k_1 +1)^2 )$. Then $ E_n\in \Upsilon
_n$. Note that at most one eigenvalue can be localized between
$k^2$ and $(k+1)^2$.
\\
In the following we denote $z\mapsto
\sqrt{z-k^2}_{\mathit{I}}\equiv \sqrt{z-k^2}$ and $z\mapsto
\sqrt{z-k^2}_{\mathit{II}}$ for, respectively, the first and the
second Riemann sheet of the square root function with the cut
$[k^2\,,\infty )$. Then, for any $k\in \N$ the function $\C_+ \cup
\Upsilon_n \ni z\mapsto \sqrt{z-k^2}_{\mathit{I}}$ has an
analytical continuation via $\Upsilon_n $ to a bounded open set
$\Upsilon_{-,n}\equiv  \Upsilon_{-} \subset \C_-$ with boundaries
containing $\Upsilon_n$. This continuation is given by
$$
\Upsilon_n \cup \Upsilon_- \ni z\mapsto  \tau _k (z)=
\begin{cases} \sqrt{z-k^2}_I &\mbox{for } k\in
\mathcal{A}_n^c\\
\sqrt{z-k^2}_{\mathit{II}}  & \mbox{for }  k\in \mathcal{A}_n\,.
\end{cases}
$$
It is clear from the definition that $\tau_k (\cdot)$ depends on
$n$.  By means of $\tau_k (\cdot )$ we can construct an analytical
continuation of $\mathrm{R}_{\epsilon }(z)$ to $\Upsilon_n \cup
\Upsilon_- $ which is given by an
 integral operator $\mathrm{R}^{\mathit{II}}_\epsilon (z)$ acting in $L^2
 ([0,\pi] )$ with the kernel
\begin{equation}\label{eq-G}
\mathrm{G}^{\mathit{II}}_\epsilon (z; x_2,y_2):=\frac{i}{\pi }
\sum_{k=1}^{\infty} \frac{\mathrm{e}^{i\tau _k (z)\epsilon |x_2
-y_2|}}{\tau_k(z)} \omega_k (x_2 )\omega_k (y_2) \,.
\end{equation}
\begin{lemma}\label{le-bounded} Suppose $z\in \Upsilon_n \cup
\Upsilon_-$. Operator $\mathrm{R}^{\mathit{II}}_\epsilon (z)$ is
bounded.
\end{lemma}
\begin{proof} Note that (\ref{eq-G}) consists of  the finite number
$\sharp \mathcal{A}_n$ of bounded operators for which $\Im \tau_k
(z) \leq 0 $. The remaining part
$$
\frac{i}{\pi} \sum _{k\in \mathcal{A}^c_n }
\frac{\mathrm{e}^{i\sqrt{z-k^2}_{I}\epsilon |x_2
-y_2|}}{\sqrt{z-k^2}_{I}} \omega_k (x_2 )\omega_k (y_2)\,
$$
is also bounded due the analogous  argument which implies
boundedness of (\ref{eq-Gepsilon}).
\end{proof}

Finally, in the analogous way as
$\mathrm{R}^{\mathit{II}}_\epsilon (z)$ we can construct the
second sheet continuation of the  remaining ingredients of the
resolvent $R_0 (\cdot)$, $\tilde{R}^\ast _{\epsilon}(\cdot)$ and
$\tilde{R}_{\epsilon}(\cdot)$. Relaying on (\ref{eq-resolvent})
 can build up  the second sheet continuation
$R^{\mathit{II}}_{\alpha , \epsilon} (z)$ which for any $f\,,g \in
C_0^\infty (\R^2)$ defines the analytic function $z\mapsto (f\,,
R^{\mathit{II}}_{\alpha , \epsilon} (z) g)$.

\subsection{Spectral condition for the resonances}

\noindent \emph{Decomposition of $\mathrm{R}_\epsilon
^{\mathit{II}}(z)$.} The goal  of this sequence is to formulate a
spectral condition for the resolvent pole located near energy
$E_{n}$. For this aim we will extract from (\ref{eq-G}) the $n$-th
component of $\mathrm{G}_0 ^{\mathit{II}} (z,\cdot,\cdot )$
``responsible'' for the existence embedded eigenvalue $E_n$ and
given by
\begin{equation}\label{eq-S}
\mathrm{S}_n (z):=\frac{i}{\pi }\frac{1}{\tau_{n} (z)}(\omega _{n}
\,, \cdot ) \omega_{n}
\,.
\end{equation}
Denote
\begin{equation}\label{eq-T}
 \mathrm{T}_\epsilon (z):=
\mathrm{R}^{\mathit{II}}_\epsilon (z) - \mathrm{S}_n (z)\,.
\end{equation}
In fact, $\mathrm{T}_\epsilon (z)$ depends also on $n$; without a
danger of confusion we  omit index $n$ here. Note that in view of
Lemma~\ref{le-bounded} operator $ \mathrm{T}_\epsilon (z)$ is
bounded. The next step is to derive expansion of $\mathrm{T}
_\epsilon (z)$ with the respect to $\epsilon$ up to the second
order perturbation term. Employing (\ref{eq-G}) we get
\begin{equation}\label{eq-expTfull}
\mathrm{T}_\epsilon (z) = \mathrm{T}_0 (z)+\mathrm{T}_1 (z)
\epsilon   +\mathrm{T}_2 (z) \epsilon^2+
\mathcal{O}(\epsilon^3)\,,
\end{equation}
where
\begin{equation}\label{eq-defT0}
\mathrm{T}_0(z) =\frac{i}{\pi} \sum _{k\neq n } \frac{1}{\tau_k
(z)} (\omega _{n} \,, \cdot )
\omega_{n}\,,
\end{equation}
and the kernels  of $\mathrm{T}_m (z)$, $m=1\,,2$  are given by
\begin{equation}\label{eq-expT}
\mathrm{T}_m (z) (x_2\,,y_2)=-\frac{i^{m-1}}{m\pi }\sum_{k\in \N}
\tau _k (z)^{m-1}  |x_2 -y_2 |^{m} \omega _k(x_2) \omega
_k(y_2) 
\,.
\end{equation}
Since both $\mathrm{T}_\epsilon (z)$ and $\mathrm{T}_0(z)$ are
bounded the ``perturbant'' $\mathrm{T}_\epsilon (z)-\mathrm{T}_0
(z)$ is bounded as well. Moreover, note that $\mathrm{T}_1
(\cdot)$ does not depend on $z$. Therefore, in the following we
will write $\mathrm{T}_1$.

\begin{lemma} \label{le-Tinv}
Suppose that $\mathcal{U}_n \in \C_+ \cup \Upsilon_n \cup
\Upsilon_- $ defines a small neighborhood of $E_n$ and $z\in
\mathcal{U}_n$. For $\epsilon $ sufficiently small the operator $I
-\alpha \mathrm{T}_{\epsilon} (z)$ is invertible.
\end{lemma}
\begin{proof}
Note that for $z\in  \mathcal{U}_n $ the function $z\mapsto \tau_k
(z)$ is analytic.  Using the explicit form of $\mathrm{T}_0
(z)$,~see~(\ref{eq-defT0}), we conclude that its eigenvalues take
the form
$$
t_k (z)=
\begin{cases} \frac{i}{2\tau_k (z)} &\mbox{for } k \neq n\\
0  & \mbox{for }  k=n\,.
\end{cases}
$$
The corresponding eigenfunctions are determined by $\omega _k$ for
$k\in \N$. Combining the above statement together with
(\ref{eq-expTfull}) and results of the perturbation theory,
\cite{Kato},   we get that the eigenvalues of $\alpha
\mathrm{T}_\epsilon (z)$, $z\in \mathcal{U}_n$ take the forms
$$
\alpha t_k (z)
+\mathcal{O}(\epsilon )\,,\,\,\,\, k\in \N\,.
$$
Since the equation
$$
\alpha t_k (z)
+\mathcal{O}(\epsilon ) =1
$$
has no solution for $z\in \mathcal{U}_n$ we can conclude that
$\ker (I-\alpha \mathrm{T}_\epsilon (z))=\{0\} $.
\end{proof}
Relaying on the above lemma we define
\begin{equation}\label{eq-defeta}
 \eta _n (z,\epsilon ):= \tau_{n } (z)-i\frac{\alpha }{\pi
} (\omega _{n}, (I-\alpha \mathrm{T}_{\epsilon}(z))^{-1}\omega
_{n})\,.
\end{equation}

\begin{theorem} \label{th-spectral}
Suppose that $z\in \mathcal{U}_n$. Then
$$
\ker{\Gamma }^{\mathit{II}}_\epsilon (z) \neq \{0\}\quad
\mathrm{iff} \quad \eta_n (\epsilon, z)=0\,.
$$
\end{theorem}
\begin{proof}
Using  decomposition (\ref{eq-T}) and Lemma~\ref{le-Tinv} we
conclude that for $z\in \mathcal{U}_n$ we have
\begin{equation}\label{eq-}
  \Gamma_\epsilon ^{\mathit{II}}(z) = (I-\alpha \mathrm{T}_{\epsilon} (z)
  )(I-\alpha (I-\alpha \mathrm{T}_{\epsilon }(z))^{-1} \mathrm{S}_n (z)
  )\,.
\end{equation}
Using again the fact that  $\ker  (I-\alpha
\mathrm{T}_{\epsilon}(z))=\{0\}$ we arrive at
$$
\ker   \Gamma_\epsilon ^{\mathit{II}}(z)\neq \{  0 \} \,\,
\Longleftrightarrow \,\, \ker (I-\alpha (I-\alpha
\mathrm{T}_{\epsilon }(z))^{-1} \mathrm{S}_n (z) ) \neq \{0 \}\,.
$$
Applying (\ref{eq-S}) we come to the conclusion that $f\in \ker
(I-\alpha (I-\alpha \mathrm{T}_{\epsilon }(z))^{-1} \mathrm{S}_n
(z) )$ $iff$
$$
f= i\frac{\alpha }{\pi \tau_{n } (z) } (\omega _{n}, f) (I-\alpha
\mathrm{T}_{\epsilon}(z))^{-1}\omega _n\,.
$$
The latter is equivalent to
$$
\eta_n (z, \epsilon ) =0 \,;
$$
this completes the proof.
\end{proof}

\subsection{Solution of the spectral equation}

Once we have the spectral equation our aim is to recover its
solutions which reproduce the resolvent poles. Since the functions
$\mathcal{U}_n \ni z \mapsto \tau_k (z)$, $k\in \N$  are analytic
the spectral condition obtained in Theorem~\ref{th-spectral} is
equivalent to
\begin{equation}\label{eq-spectral2}
\tau_n (z)^2 = -\frac{\alpha^2}{\pi^2} (\omega_n , (I-\alpha
\mathrm{T}_\epsilon (z))^{-1}\omega_n)^2\,,
\end{equation}
cf.~(\ref{eq-defeta}).
\\
\emph{Expansion of $(\omega_n , (I-\alpha \mathrm{T}_\epsilon
(z))^{-1}\omega_n)$. } The next step is to expand the expression
$(\omega_n , (I-\alpha \mathrm{T}_\epsilon (z))^{-1}\omega_n)$
involved in the above equation with respect to $\epsilon $.
Suppose $A$ is invertible with the bounded inverse and $C$ is
bounded with the norm $\|C\|=\mathcal{O}(\epsilon)$. Then
\begin{equation}\label{eq-expA+C}
(A- C)^{-1} = A^{-1}(I+A^{-1}C+(A^{-1}C)^2+... )\,.
\end{equation}
Specify $A:=I-\alpha \mathrm{T}_0  (z)$ and $C:=\alpha (
 \mathrm{T}_\epsilon (z)-\mathrm{T}_0  (z)) $. Combining
 (\ref{eq-expA+C}) and
(\ref{eq-expT}) together with the fact $A\omega _n =\omega _n$ we
get
\begin{eqnarray} \nonumber
&&(\omega_n , (I-\alpha \mathrm{T}_\epsilon (z))^{-1}\omega_n)=
(\omega_n , (I+C+C^2+...)\omega_n)=\\ \nonumber && \frac{\pi }{2}+
\alpha \epsilon (\omega_n , \mathrm{T}_1 \omega_n) + \alpha
\epsilon ^2 (\omega_n , \mathrm{T}_2 (z)\omega_n)
\\
\label{eq-expomega} && + \alpha^2 \epsilon^2
 (\omega_n , \mathrm{T}_1 ^2 \omega_n)
 +\mathcal{O}(\epsilon^3)\,.
\end{eqnarray}
With the above statements we are ready to prove the main theorem.
\begin{theorem} \label{th-main}
Let $\alpha >0$. Suppose $n^2>1+\frac{\alpha ^2}{4}$, i.e. the
number $E_n =-\frac{\alpha ^2}{4} +n^2$ determines the embedded
eigenvalue of $H_{\alpha ,0}$. Then the Hamiltonian $H_{\alpha
,\epsilon}$ has the resolvent pole at
\begin{equation}\label{eq-exppole}
z_n (\epsilon )= E_n+V_n  \epsilon +W_n
\epsilon^2+\mathcal{O}(\epsilon^3)\,,
\end{equation}
where
$$
V_n :=-\frac{\alpha ^3}{\pi } (\omega_n , \mathrm{T}_1
\omega_n)\,,
$$
$$
W_n  :=-\frac{\alpha ^4 }{\pi^2} (\omega_n , \mathrm{T}_1 \omega_n
)^2 -  \frac{\alpha ^3 }{\pi} (\omega_n , \mathrm{T}_2
(E_n)\omega_n) - \frac{\alpha ^4 }{\pi} (\omega_n , \mathrm{T}_1
^2 \omega_n)\,.
$$
\end{theorem}
 \begin{proof} Suppose $z\in \mathcal{U}_n $.
 The spectral condition (\ref{eq-spectral2}) reads as
\begin{equation}\label{eq-spectral3}
  \tilde{\eta}_n (z,\epsilon)= 0 \,,
\end{equation}
 where
$$
\tilde{\eta}_n (z,\epsilon) := z-n^2 -\frac{\alpha^2}{\pi^2}
(\omega_n , (I-\alpha \mathrm{T}_\epsilon (z))^{-1}\omega_n)^2\,.
$$
Using the expansion (\ref{eq-expomega}) we get
\begin{eqnarray} \nonumber
\tilde{\eta}_n (z,\epsilon) &&=z-n^2 -\frac{\alpha^2}{4}
-\frac{\alpha
^3}{\pi } (\omega_n , \mathrm{T}_1 \omega_n )\epsilon +\\
\nonumber && \left(-\frac{\alpha ^4 }{\pi^2} (\omega_n ,
\mathrm{T}_1 \omega_n )^2 -  \frac{\alpha ^3 }{\pi} (\omega_n ,
\mathrm{T}_2 (z)\omega_n) - \frac{\alpha ^4 }{\pi} (\omega_n ,
\mathrm{T}_1 ^2 \omega_n) \right)
\epsilon^2+\mathcal{O}(\epsilon^3)\,.
\end{eqnarray}
Note that
$$
\tilde{\eta}_n (E_n ,0 )= 0\,.
$$
Since $\mathcal{U}_n \ni z \mapsto  \tilde{\eta}_n (z\,,\epsilon
)$ is analytic (this follows from the analycity of
$\mathrm{R}_\epsilon ^{\mathit{II}}(\cdot )$) and $ \tilde{\eta}_n
(z\,,\epsilon )$ is $C^\infty$ as the function of $\epsilon$ we
can apply the implicit function theorem in view of which the
equation (\ref{eq-spectral3}) has a unique solution in the
neighborhood $\mathcal{U}_n$ of $E_n$. The expansion
(\ref{eq-exppole}) comes directly from the above expansion of
$\tilde{\eta}_n (z,\epsilon)$ .
 \end{proof}

\section{Geometrically induced resonances: discussion}

\subsection{Imaginary part of the resolvent pole}

Theorem~\ref{th-main} shows that the pole of the resolvent
dissolving originally in the essential spectrum is  slightly
shifted after breaking the symmetry. Since the operator
$H_{\alpha, \epsilon }$ is self-adjoint the pole of its resolvent
can other stay in the essential spectrum or turn to the lower
complex half plane. The aim of this section is to show that the
later holds for $\epsilon $ small enough. The question is
appealing because the imaginary part of the pole determines the
width of resonance.
\\
Note that the error term in (\ref{eq-exppole}) depends on the
index $n$; i.e. for $n\in \N$ there exists $\epsilon_n
>0$ such that for any $\epsilon \in (0,\epsilon_n)$ the term $W_n
\epsilon ^2 $ dominates w.r.t. the error term (in the sense of the
real and imaginary components).
\\
Fix $n\in\N$ and suppose that $\epsilon \in (0,\epsilon_n)$.  Note
that $\mathrm{T}_1 $ is a self-adjoint operator and, consequently,
$V_n$ is real; also the first and third component of $W_n$ (see
(\ref{eq-exppole})) are real. A nontrivial imaginary component of
$z_n (\epsilon)$ of the lowest order is induced by
\begin{eqnarray} \label{eq-imz}
&&(\omega_n , \mathrm{T}_2 (E_n)\omega_n)=
-\frac{i}{2\pi } \sum_{k\in\N} \tau_k (E_n ) N_{k,n}\,,
\end{eqnarray}
where
\begin{eqnarray} \nonumber
N_{k,n}:= && \int_{0}^{\pi} \int_{0}^{\pi} |x_2-y_2|^2 \sin (k
x_2)\sin (n x_2)\sin (k y_2) \sin (n y_2) \mathrm{d}x_2
\mathrm{d}x_y= \\ \nonumber &&
\begin{cases}  
-8 \frac{k^2n^2}{(n^2-k^2)^4} \quad & \quad \mbox{if} \,\, |k-n| \,\, \mbox{odd} \\
0 \quad & \quad \mbox{if} \,\,  |k-n| \,\, \mbox{even}\\
\frac{\pi^4}{12}-\frac{\pi^2}{4 k^2} \quad & \quad \mbox{if} \,\,
k=n \,.
\end{cases}
\end{eqnarray}
Note that not all components of (\ref{eq-imz}) contribute to
imaginary part of $z_n (\epsilon)$. Precisely, $\tau_k (E_n)$  is
purely real for $k\in \mathcal{A}_n $. Otherwise, for $k\in
\mathcal{A}_n^c$ it is purely imaginary and, consequently, the
corresponding components are \emph{not} employed in the imaginary
part of the resonance pole. This, in view of (\ref{eq-imz})
implies
\begin{equation}\label{eq-imaginary}
\Im z_n (\epsilon)= -\frac{4\alpha ^2}{\pi ^3}S_n \epsilon^2
 +\mathcal{O}(\epsilon^3)\,,
\end{equation}
where
$$
S_n :=\sum_{k\in \mathcal{A}'_n} \tau_k (E_n )
\frac{k^2n^2}{(n^2-k^2)^4}\,
$$
and  $\mathcal{A}'_n := \{k\in \mathcal{A}_n \,:\, n-k\,\,
\mathrm{is}\,\, \mathrm{odd}\}$. Since  $\tau_k (E_n ) >0$ for all
$k\in \mathcal{A}'_n$ we obtain $\Im z_n (\epsilon) <0$.
\\
Formula (\ref{eq-imaginary})  explicitly reveals the lowest order
of the resonance width
$$ \gamma_n =  \frac{2\alpha ^2}{\pi ^3} S_n \epsilon^2\,. $$

\subsection{The resonance energy }
The real component of the resonance pole determines  the resonance
energy. Fix $n\in \N$ and suppose $\epsilon $ is small enough. The
lowest order of the real component associated with $z_n
(\epsilon)$ is given by
$$
E_n -\frac{\alpha^3}{\pi} (\omega_n , \mathrm{T}_1 \omega _n
)\epsilon\,.
$$
To derive more explicit expression we calculate
$$
(\omega_n , \mathrm{T}_1 \omega _n ) =  \sum_{k\in \N} M_{k,n}\,,
$$
where
\begin{eqnarray} \nonumber
M_{k,n}:= &&\int_{0}^{\pi} \int_{0}^{\pi} |x_2-y_2| \sin (k
x_2)\sin (n
x_2)\sin (k y_2) \sin (n y_2) \mathrm{d}x_2 \mathrm{d}x_y= \\
\nonumber &&
\begin{cases}  
- \frac{\pi}{4}\frac{n^2+k^2 }{(n^2-k^2)^2}\quad & \quad \mbox{if}
\,\, k\ \neq n
 \\
-\frac{21}{48}\frac{\pi}{k^2}+\frac{11}{6}\pi^3 \quad & \quad
\mbox{if} \,\,  k=n  \,.
\end{cases}
\end{eqnarray}

\subsection{Comments on repulsive potential}

Suppose that the coupling constant $\alpha <0$; i.e. the delta
potential localized on $\Sigma _\epsilon$ has a repulsive
character. Moreover, assume at the beginning that $\epsilon = 0 $.
Note, that in this case, equation (\ref{eq-BSII}) do not admit any
solution (remind that $\Im \sqrt{z-k^2} >0$) and, consequently,
embedded eigenvalues do not appear. However, for $\alpha <0$ the
analogous equation but for the second sheet continuation of the
square root function, i.e.
$$
2 \sqrt{z-k^2}_{\mathit{II}}-i\alpha =0\,
$$
has a solution at $z_k =-\frac{\alpha^2}{4}+k^2 $. Let us
emphasize that the above solution does not determine an eigenvalue
of $H_{\alpha ,0}$. It  states a pole of the resolvent  living on
second sheet. The analogous phenomena occurs for the repulsive one
point interaction in one dimensional system,
cf.~\cite{AGHH},~Chap.~I.3. If we break the symmetry we may expect
that the pole moves slightly but since $H_{\alpha , \epsilon }$ is
a self-adjoint operator the pole has to stay on the second sheet.
In this case it means that its imaginary part is nonnegative.
Indeed, repeating the procedure employed for $\alpha
>0 $ after obvious changes we get
$$
\Im \tilde{z}_n (\epsilon)= -\frac{4\alpha ^2}{\pi ^3}\tilde{S}_n
\epsilon^2
 +\mathcal{O}(\epsilon^3)\,,
$$
where
$$
\tilde{S}_n :=\sum_{k\in \mathcal{A}'_n} \sqrt{E_n +k^2}
_{\mathit{II}} \,\frac{k^2n^2}{(n^2-k^2)^4}\,.
$$
Since $  \sqrt{E_n -k^2} _{\mathit{II}} <0$ we have $\Im
\tilde{z}_n (\epsilon) >0$ for $\epsilon $ small enough.

\subsection{Open questions}

The first natural question arising here concerns the higher
dimensional models. For example, suppose that the geometry of the
waveguide is is defined by $\Omega := \{ (x_1\,,x_2\,,x_3)\,:\,
x_1\,,x_2\in \R\,,x_3\in (0, \pi)\} $ and delta interaction is
supported by $\Sigma_\epsilon :=\{(0, \epsilon x_3, x_3)\}
\subset\Omega $. Since the interaction is supported by a set of
codimension two the model is essentially different from the one
studied in this paper. The resolvent needs a certain kind of the
renormalization and this makes the resonance analysis more
involved.
\\ \\
The other interesting question concerns the spectral properties of
the system if   $\epsilon$ is large. Especially,  if $\epsilon $
goes formally to infinity  then $\Sigma_\epsilon$ is getting
parallel to the walls of $\Omega$. This suggests that the
eigenvalues of $H_{\alpha, \epsilon }$ are  densely localized and
at limiting case we get an  additional component of the essential
spectrum.



\subsection*{Acknowledgement}
The author thanks the referees for reading the paper and
recommending various improvements in exposition.

\end{document}